\documentclass[runningheads,a4paper]{llncs}

\usepackage{mathrsfs}
\usepackage[USenglish]{babel}
\usepackage{theoremref} 
\usepackage{extarrows}
\usepackage{color}
\usepackage{xspace}
\usepackage{tikz}
	\usetikzlibrary{arrows,shapes,automata,backgrounds,decorations}
\usepackage{graphicx}
\usepackage{mdwlist,paralist}
\usepackage{url}
\usepackage{xifthen}
\usepackage{hyperref}
\usepackage{DynamicCausality}

\usepackage{wasysym,amsmath,mathtools,amssymb}

\begin{document}

\title{Dynamic Causality in Event Structures (Technical Report)\thanks{Supported by the DFG Research Training Group SOAMED.}}
\author{Youssef Arbach, David Karcher, Kirstin Peters, Uwe Nestmann}
\institute{Technische Universit\"at Berlin, Germany\\
\texttt{\{youssef.arbach, david.s.karcher, kirstin.peters, uwe.nestmann\}@tu-berlin.de}}
\maketitle

\begin{abstract}
	In \cite{dynamicCausality15} we present an extension of Prime Event Structures by a mechanism to express dynamicity in the causal relation. More precisely we add the possibility that the occurrence of an event can add or remove causal dependencies between events and analyse the expressive power of the resulting Event Structures \wrt to some well-known Event Structures from the literature. This technical report contains some additional information and the missing proofs of \cite{dynamicCausality15}.
\end{abstract}

\setcounter{definition}{19}
\setcounter{lemma}{1}
\setcounter{theorem}{10}

\begin{figure}[t]
	\centering
	\begin{tikzpicture}[bend angle=30]
		% figure a
		%\event{a1}{0}{2}{left}{$ a $};
		\event{b2}{0.3}{1.0}{left}{$ a $};
		\event{b3}{1.3}{1.0}{right}{$ b $};
		\event{b4}{0.8}{0}{left}{$ c $};
		%\draw[enablingPES] (a1) edge (a2);
		\draw[enablingPES] (b2) edge (b4);
		\draw[enablingPES] (b3) edge (b4);
		\draw[thick] (0.55, 0.5) -- (1.05, 0.5);
		\draw[conflictPES] (b2) edge (b3);
		\node (g) at (-0.5, 1) {$(\beta_\gamma)$};
		
		% figure b
		\event{a1}{3}{1}{left}{$ a $};
		\event{a2}{4}{1}{right}{$ b $};
		\event{a3}{3}{0}{left}{$ c $};
		\node (a) at (2.2, 1) {$ (\sigma_\xi) $};
		\draw[enablingPES] (a1) edge (a2);
		\draw[dropping]	(3.4, 1) -- (a3);
		
		% figure
		\event{b1}{5.6}{0.5}{above}{$ e $};
		\event{b2}{6.6}{0.5}{above}{$ f $};
		\node (b) at (5, 1) {$ (\xi_\sigma) $};
		\draw[conflictEBES] (b1) edge (b2);
		
		% figure
		\event{ba}{8.3}{1}{above right}{$ b $};
		\event{bb}{8.3}{0}{below left}{$ a $};
		%\node (b) at (1.3, 1) {(b)};
		\node (gs) at (7.5, 1) {$(\gamma_\sigma)$};
		\draw[enablingAbsent] (bb) .. controls (8.8, -0.5) and (8.8, 0.5) .. (bb);
		\draw[adding] (ba) .. controls (8.6, 1) .. (8.6, 0.3);
		
		% figure c
		\event{ca}{9.8}{0}{below left}{$ b $};
		\event{cb}{10.8}{0}{below right}{$ c $};
		\event{cc}{10.3}{1}{above}{$ a $};
		%\node (c) at (3.2, 1) {(c)};
		\node (gx) at (9.4, 1) {$(\gamma_\xi)$};
		\draw[enablingAbsent] (ca) -- (cb);
		\draw[adding] (cc) -- (10.3, 0);
	\end{tikzpicture}
	\caption{Counterexamples.}
	\label{fig:counterExamples}
\end{figure}
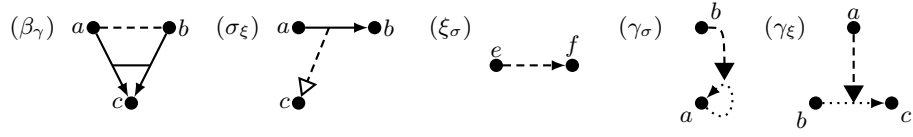

\section{Event Structures for Resolvable Conflict}

For a transition based ES with a few additional properties, there is a natural
embedding into RCESs.
\begin{definition}\label{def:translationIntoRCES}
	Let $ \mu $ be an ES with a transition relation $ \transOp $ defined on configurations such that
	%\begin{compactitem}
		%\item 
		$ \trans{X}{Y} $ implies $ X \subseteq Y $ and
		%\item 
		$ X \subseteq X' \subseteq Y' \subseteq Y $ implies that $ \trans{X}{Y} \implies \trans{X'}{Y'} $
	%\end{compactitem}
	for all configurations $ X, Y, X', Y' $ of $ \mu $.	Then $ \rces{\mu} = \left( E, \Set{ \enrc{X}{Z} \mid \exists Y \subseteq E \logdot \trans{X}{Y} \land Z \subseteq Y } \right) $.
\end{definition}
Note that the SESs and GESs satisfy these properties. We show that the resulting
structure $ \rces{\mu} $ is indeed an RCES and that it is transition equivalent to $ \mu $.
\begin{lemma}	\label{lma:TransInRCES}
	Let $ \mu $ be an ES that satisfies the conditions of \defi\ref{def:translationIntoRCES}.	Then $ \rces{\mu} $ is a RCES and $  \transEq{\rces{\mu}}{\mu}$.
\end{lemma}
\begin{proof}
	By Def. 9 in \cite{dynamicCausality15}, $ \rces{\mu} $ is a RCES.
	
	Assume $ \trans{X}{Y} $. Then, by \defi\ref{def:translationIntoRCES}, $ X \subseteq Y $ and $ \enrc{X}{Z} $ for all $ Z \subseteq Y $.Then, by Def. 10 in \cite{dynamicCausality15}, $ \transRC{X}{Y} $.
	
	Assume $ \transRC{X}{Y} $. Then, by Def. 10 in \cite{dynamicCausality15}, $ X \subseteq Y $ and there is some $ X' \subseteq X $ such that $ \enrc{X'}{Y} $. By \defi\ref{def:translationIntoRCES} for $ \rces{\cdot} $, then $ \enrc{X'}{Y'} $ for all $ Y' \subseteq Y $. So there is a set $ \widetilde{Y} $ such that $ Y \subseteq \widetilde{Y} $ and $ \enrc{X'}{\widetilde{Y}'} $ for each $ \widetilde{Y}' \subseteq \widetilde{Y} $. Then, by \defi\ref{def:translationIntoRCES} for $ \rces{\cdot} $, it follows $ \trans{X'}{\widetilde{Y}'} $ and $ X' \subseteq X \subseteq Y \subseteq \widetilde{Y} $. Finally, by the second property of \defi\ref{def:translationIntoRCES}, $ \trans{X}{Y} $.
\end{proof}

\section{Shrinking Causality}

In SESs both notions of configurations, traced-based and transition-based,
coincide; and in different situations, the more suitable one can be used.

\begin{lemma}\label{lem:SESconf}
	Let $ \sigma $ be a SES. Then $ \configTraces{\sigma} = \configurations{\sigma} $.
\end{lemma}

\begin{proof}
	Let $ \sigma = \left( E, \confOp, \enabOp, \shrinkingCausality \right) $.	By Def. 13 in \cite{dynamicCausality15}, $ C \in \configTraces{\sigma} $ implies that there is some $ t = e_1 \cdots e_n $ such that $ \overline{t} \subseteq E $, $ \forall 1 \leq i, j \leq n \logdot \neg \left( \conf{e_i}{e_j} \right) $, $ \forall 1 \leq i \leq n \logdot \left( \ic{e_i} \setminus \dc{\overline{t_{i - 1}}}{e_i} \right) \subseteq \overline{t_{i - 1}} $, and $ C = \overline{t} $.	Hence, by Def. 13, $ \transS{\overline{t_i}}{\overline{t_{i + 1}}} $ for all $ 1 \leq i \leq n $ and $ \transS{\emptyset}{\Set{ e_1 }} $.	Thus, by Def. 13, $ C \in \configurations{\sigma} $.
	
	By Def. 13, $ C \in \configurations{\sigma} $ implies that there are $ X_1, \ldots, X_n \subseteq E $ such that $ \transS{\transS{\transS{\emptyset}{X_1}}{\ldots}}{X_n} $ and $ X_n = C $.
	Then, by Def. 13, we have:
	\begin{align}
		& \emptyset \subseteq X_1 \subseteq X_2 \subseteq \ldots \subseteq X_n \subseteq E & \tag{C1} \label{eq:C1}\\
		& \forall e, e' \in X_n \logdot \neg \left( \conf{e}{e'} \right) & \tag{C2} \label{eq:C2}\\
		& \forall e \in X_1 \logdot \left( \ic{e} \setminus \dc{\emptyset}{e} \right) \subseteq \emptyset & \tag{C3} \label{eq:C3}\\
		& \forall 1 \leq i < n \logdot \forall e \in X_{i + 1} \setminus X_i \logdot \left( \ic{e} \setminus \dc{X_i}{e} \right) \subseteq X_i & \tag{C4} \label{eq:C4}
	\end{align}
	Let $ X_1 = \Set{ e_{1, 1}, \ldots, e_{1, m_1} } $ and $ X_i \setminus X_{i - 1} = \Set{ e_{i, 1}, \ldots, e_{i, m_i} } $ for all $ 1 < i \leq n $.	Then, by Def. 13, $ t = e_{1, 1} \cdots e_{1, m_1} \cdots e_{n, 1} \cdots e_{n, m_n} = e_1' \cdots e_k' $ is a trace such that $ \overline{t} \subseteq E $ (because of \eqref{eq:C1}), $ \neg \left( \conf{e_i'}{e_j'} \right) $ for all $ 1 \leq i, j \leq k $ (because of \eqref{eq:C1} and \eqref{eq:C2}), for all $ 1 \leq i \leq k $ and all $ 1 \leq j \leq m_i $ we have $ \left( \ic{e_{i, j}} \setminus \dc{\overline{t_{i - 1}}}{e_{i, j}} \right) \subseteq \overline{t_{i - 1}} $ (because of \eqref{eq:C3} and \eqref{eq:C4}), and $ \overline{t} = C $ (because $ X_n = C $).
	Thus $ C \in \configTraces{\sigma} $.
\end{proof}

Moreover the following technical Lemma relates transitions and the extension of traces by causally independent events.

\begin{lemma}
	\label{lem:SEStransTraces}
	Let $ \sigma = \left( E, \confOp, \enabOp, \shrinkingCausality \right) $ be a SES and $ X, Y \in \configurations{\sigma} $.	Then $ \transS{X}{Y} $ iff there are $ t_1 = e_1 \cdots e_n, t_2 = e_1 \cdots e_n e_{n + 1} \cdots e_{n + m} \in \traces{\sigma} $ such that $ X = \overline{t_1} $, $ Y = \overline{t_2} $, and $ \forall e, e' \in Y \setminus X \logdot \left( \ic{e} \setminus \dc{X}{e} \right) \subseteq X $.
\end{lemma}

\begin{proof}
	By Def. 13 in \cite{dynamicCausality15} and \lem\ref{lem:SESconf}, $ X \in \configurations{\sigma} $ implies that there is a trace $ t_1 = e_1 \cdots e_n \in \traces{\sigma} $ such that $ X = \overline{t_1} $.

	If $ \transS{X}{Y} $ then, by Def. 13, $ X \subseteq Y $, $ \forall e, e' \in Y \logdot \neg ( \conf{e}{e'} ) $, and $ \forall e \in Y \setminus X \logdot ( \ic{e} \setminus \dc{X}{e} ) \subseteq X $. Then, by Def. 13, $ t_2 = e_1 \cdots e_n e_{n + 1} \cdots e_{n + m} \in \traces{\sigma} $ and $ Y = \overline{t_2} $ for an arbitrary linearization $ e_{n + 1} \cdots e_{n + m} $ of the events in $ Y \setminus X $, \ie with $ \Set{ e_{n + 1}, \ldots, e_{n + m} } = Y \setminus X $ such that $ e_{n + i} \neq e_{n + j} $ whenever $ 1 \leq i, j, \leq m $ and $ i \neq j $.
	
	If there is a trace $ t_2 = e_1 \cdots e_n e_{n + 1} \cdots e_{n + m} \in \traces{\sigma} $ such that $ Y = \overline{t_2} $ and $ \forall e, e' \in Y \setminus X \logdot \left( \ic{e} \setminus \dc{X}{e} \right) \subseteq X $ then $ X \subseteq Y $. Moreover, by Def. 13, $ t_2 \in \traces{\sigma} $ implies $ \forall e, e' \in Y \logdot \neg \left( \conf{e}{e'} \right) $. Thus, by Def. 13, $ \transS{X}{Y} $.
\end{proof}

Note that the condition $ \forall e, e' \in Y \setminus X \logdot \left( \ic{e}
\setminus \dc{X}{e} \right) \subseteq X $ states that the events in $ Y
\setminus X $ are causally independent from each other.

As mentioned above DESs and SESs have the same expressive power. To show this fact we define mutual encodings and show that they result into structures with equivalent behaviors. To translate a SES into a DES we create a bundle for each initial causal dependence and add all its droppers to the the bundle set.

\begin{definition}
	\label{def:SESintoDES}
	Let $ \sigma = \left( E, \confOp, \enabOp, \shrinkingCausality \right) $ be a SES.
	Then $ \des{\sigma} = \left( E, \confOp, \buEnOp \right) $, where $ \buEn{S}{y} $ iff $ S \subseteq E $, $ y \in E $, and $ \exists x \in E \logdot \enab{x}{y} \land S = \Set{x} \cup \droppers{x}{y} $.
\end{definition}

The above translation from SES into DES shows that for each SES there is a DES with exactly the same traces and configurations.

\begin{lemma}
	\label{lem:SEStoDES}
	For each SES $ \sigma $ there is a DES $ \delta $, namely $ \delta = \des{\sigma} $, such that $ \traces{\sigma} = \traces{\delta} $ and $ \configurations{\sigma} = \configurations{\delta} $.
\end{lemma}

\begin{proof}[Proof  of Lemma~\ref{lem:SEStoDES}]
	Let $ \sigma = \left( E, \confOp, \enabOp, \shrinkingCausality \right) $ be a SES. By \defs12 and 1 in \cite{dynamicCausality15}, $ \confOp \subseteq E^2 $ is irreflexive and symmetric. Hence, by \defs{7 in \cite{dynamicCausality15}} and \ref{def:SESintoDES}, $ \delta = \des{\sigma} $ is a DES.
	
	Let $ t = e_1 \cdots e_n $.
	By Def. 13 in \cite{dynamicCausality15}, $ t \in \traces{\sigma} $ iff $ \overline{t} \subseteq E $, $ \neg \left( \conf{e_i}{e_j} \right) $, and $ ( \ic{e_i} \setminus \dc{\overline{t_{i - 1}}}{e_i} ) \subseteq \overline{t_{i - 1}} $ for all $ 1 \leq i, j \leq n $.
	Since $ \dc{H}{e} = \{ e' \mid \exists d \in H \logdot \drops{e'}{d}{e} \} $ and $ \ic{e} = \Set{ e' \mid \enab{e'}{e} } $, we have $ \left( \ic{e_i} \setminus \dc{\overline{t_{i - 1}}}{e_i} \right) \subseteq \overline{t_{i - 1}} $ iff $ \forall e' \in E \logdot \enab{e'}{e_i} \implies e' \in \overline{t_{i - 1}} \lor \exists d \in \overline{t_{i - 1}} \logdot \drops{e'}{d}{e_i} $ for all $ 1 \leq i \leq n $.
	By \defi\ref{def:SESintoDES}, then $ t \in \traces{\sigma} $ iff $ \overline{t} \subseteq E $, $ \neg \left( \conf{e_i}{e_j} \right) $, and $ \buEn{X}{e_i} \implies \overline{t_{i - 1}} \cap X \neq \emptyset $ for all $ 1 \leq i, j \leq n $ and all $ X \subseteq E $.
	Hence, by the definition of traces in \S~2.2 in \cite{dynamicCausality15}, $ t \in
	\traces{\sigma} $ iff $ t \in \traces{\delta} $, \ie $ \traces{\sigma} = \traces{\delta} $.
	
	By \lem\ref{lem:SESconf}, \S~2.2, and Def. 13, then also $ \configurations{\delta} = \configTraces{\sigma} = \configurations{\sigma} $.
\end{proof}

The most discriminating behavioral semantics of DESs used in literature are
families of posets. Thus the translation should also preserve posets.

\begin{theorem}
	\label{thm:SESintoDES}
	For each SES $ \sigma $ there is a DES $ \delta = \des{\sigma} $, such that $ \posetsEq {\sigma }{\delta} $.
\end{theorem}

\begin{proof}
	Let $ \sigma = \left( E, \confOp, \enabOp, \shrinkingCausality \right) $ be a SES.\\
	By \lem\ref{lem:SEStoDES}, $ \delta = \des{\sigma} = \left( E, \confOp, \buEnOp \right) $ is a DES such that $ \traces{\sigma} = \traces{\delta} $ and $ \configurations{\sigma} = \configurations{\delta} $.\\
	Let $ t = e_1 \cdots e_n \in \traces{\sigma} $, $ 1 \leq i \leq n $, and the bundles $ \buEn{X_1}{e_i}, \ldots, \buEn{X_m}{e_i} $ all bundles pointing to $ e_i $. For $ U $ to be a cause for $ e_i $ \defi13 in \cite{dynamicCausality15} requires $ ( \ic{e_i} \setminus \dc{U}{e_i} ) \subseteq U $. Since $ \dc{H}{e} = \Set{ e' \mid \exists d \in H \logdot \drops{e'}{d}{e} } $ and $ \ic{e} = \{e'\mid\enab{e'}{e}\} $, this condition holds iff the condition $ \enab{e'}{e_i} \implies e' \in U \lor \exists d \in U \logdot \drops{e'}{d}{e_i} $ holds for all $ e' \in E $.
	By \defi\ref{def:SESintoDES}, then $ ( \forall 1 \leq k \leq n \logdot X_k \cap U \neq \emptyset ) \iff ( ( \ic{e_i} \setminus \dc{U}{e_i} ) \subseteq U ) $.
	So, by \defs8 and 13 in \cite{dynamicCausality15}, $ \posetsEq {\sigma }{\delta}$.
\end{proof}

In the opposite direction we map each DES into a set of similar SESs such that each SES in this set has the same behavior as the DES. Therefore for each bundle $ \buEn{X_i}{e} $ we choose a fresh event $ x_i $ as initial cause $ \enab{x_i}{e} $, make it impossible by a self-loop $ \enab{x_i}{x_i} $, and add all events $ d $ of the bundle $ X_i $ as droppers \drops{x_i}{d}{e}.

\begin{definition}
	\label{def:DESintoSES}
	Let $ \delta = \left( E, \confOp, \buEnOp \right) $ be a DES, $ \Set{ X_i }_{i \in I} $ an enumeration of its bundles, and $ \Set{ x_i }_{i \in I} $ a set of fresh events, \ie $ \Set{ x_i }_{i \in I} \cap E = \emptyset $.
	Then $ \ses{\delta} = \left( E', \confOp, \enabOp, \shrinkingCausality \right) $ with $ E' = E \cup \Set{ x_i }_{i \in I} $, $ \enabOpT = \Set{ \enab{x_i}{e} \mid \buEn{X_i}{e} } \cup \Set{ \enab{x_i}{x_i} \mid i \in I } $, and $ \shrinkingCausality = \Set{ \drops{x_i}{d}{e} \mid d \in X_i \land \buEn{X_i}{e} } $.
\end{definition}

Of course it can be criticized that the translation adds events (although they are fresh and impossible). But as the following example---with more bundles than events---shows it is not always possible to translate a DES into a SES without additional impossible events.

\begin{lemma}
	\label{lem:DESninSES}
	There are DESs $ \delta = \left( E, \confOp, \buEnOp \right) $, as \eg $ \delta = \left( \Set{ a, b, c, d, e }, \emptyset, \buEnOp \right) $ with $ \buEnOp \; = \Set{ \buEn{\Set{ x, y }}{e} \mid x, y \in \Set{ a, b, c, d } \land x \neq y } $,
	that cannot be translated into a SES $ \sigma = \left( E, \confOp', \enabOp, \shrinkingCausality \right) $ such that $ \traces{\delta} = \traces{\sigma} $.
\end{lemma}

\begin{proof}[Proof of Lemma~\ref{lem:DESninSES}]
	Assume a SES $ \sigma = \left( E, \confOp, \enabOp, \shrinkingCausality \right) $ such that $ E = \{ a, b, c, d,$ $ e \} $ and $ \traces{\sigma} = \traces{\delta} $.
	According to \S~2.2 in \cite{dynamicCausality15}, $ \traces{\delta} $ contains all
	sequences of distinct events of $ E $ such that $ e $ is not the first, second, or third event, \ie for $ e $ to occur in a trace it has to be preceded by at least three of the other events.
	Since by Def. 13 in \cite{dynamicCausality15} conflicts cannot be dropped, $ \traces{\sigma} = \traces{\delta} $ implies $ \confOp = \emptyset $.
	Moreover, since $ e $ has to be preceded by at least three other events that can occur in any order, $ \enabOp $ has to contain at least three initial causes for $ e $. \WLOG let $ \enab{a}{e} $, $ \enab{b}{e} $, and $ \enab{c}{e} $.
	Because of the traces $ abd, acd \in \traces{\delta} $, we need the droppers \drops{b}{d}{e} and \drops{c}{d}{e}. Then $ ad \in \traces{\sigma} $ but $ ad \notin \traces{\delta} $.
	In fact if we fix $ E = \Set{ a, b, c, d, e } $ there only finitely many different SESs $ \sigma = \left( E, \confOp, \enabOp, \shrinkingCausality \right) $ and for none of them $ \traces{\delta} = \traces{\sigma} $ holds.
\end{proof}

Note that the above lemma implies that no translation of the above DES can result into a SES with the same events such that the DES and its translation have same configurations or posets.
However, because the $ x_i $ are fresh, there are no droppers for the self-loops $ \enab{x_i}{x_i} $ in $ \ses{\delta} $. So the translation ensures that all events in $ \Set{ x_i }_{i \in I} $ remain impossible forever in the resulting SES. In fact we show again that the DES and its translation have the exactly same traces and configurations.

\begin{lemma}
	\label{lem:DEStoSES}
	For each DES $ \delta $ there is a SES $ \sigma $, namely $ \sigma = \ses{\delta} $, such that $ \traces{\delta} = \traces{\sigma} $ and $ \configurations{\delta} = \configurations{\sigma} $.
\end{lemma}

\begin{proof}[Proof of Lemma~\ref{lem:DEStoSES}]
	Let $ \delta = \left( E, \confOp, \buEnOp \right) $ be a DES. By Def. 7 in \cite{dynamicCausality15}, $ \confOp \subseteq E^2 $ is irreflexive and symmetric. Hence, by \defs12, 1 in \cite{dynamicCausality15}, and \ref{def:DESintoSES}, $ \sigma = \ses{\delta} = \left( E', \confOp, \enabOp, \shrinkingCausality \right) $ is a SES.	
	
	Let $ t = e_1 \cdots e_n $.	Then, by Def. 13 in \cite{dynamicCausality15}, $ t \in \traces{\sigma} $ iff $ \overline{t} \subseteq E $, $ \neg ( \conf{e_i}{e_j} ) $, and $ ( \ic{e_i} \setminus \dc{\overline{t_{i - 1}}}{e_i} ) \subseteq \overline{t_{i - 1}} $ for all $ 1 \leq i, j \leq n $. Note that we have $ \overline{t} \subseteq E $ instead of $ \overline{t} \subseteq E' $, because all events in $ t $ have to be distinct and for all events in $ E' \setminus E $ there is an initial self-loop but no dropper.
	Since $ \dc{H}{e} = \Set{ e' \mid \exists d \in H \logdot \drops{e'}{d}{e} } $ and $ \ic{e} = \Set{ e' \mid \enab{e'}{e} } $, we have $ \left( \ic{e_i} \setminus \dc{\overline{t_{i - 1}}}{e_i} \right)$ $ \subseteq \overline{t_{i - 1}} $ iff $ \forall e' \in E \logdot \enab{e'}{e_i} \implies e' \in \overline{t_{i - 1}} \lor \exists d \in \overline{t_{i - 1}} \logdot \drops{e'}{d}{e_i} $ for all $ 1 \leq i \leq n $.
	By \defi\ref{def:DESintoSES}, then $ t \in \traces{\sigma} $ iff $ \overline{t} \subseteq E $, $ \neg \left( \conf{e_i}{e_j} \right) $, and $ \buEn{X}{e_i} \implies \overline{t_{i - 1}} \cap X \neq \emptyset $ for all $ 1 \leq i, j \leq n $ and all $ X \subseteq E $.
	Hence, by the definition of traces in \S~2.2 in \cite{dynamicCausality15}, $ t \in
	\traces{\sigma} $ iff $ t \in \traces{\delta} $, \ie $ \traces{\sigma} = \traces{\delta} $.
	
	By \lem\ref{lem:SESconf}, the definition of
	configurations in \S~2.2, and Def. 13, then
	also $ \configurations{\delta} = \configTraces{\sigma} = \configurations{\sigma} $.
\end{proof}

Moreover the DES and its translation have exactly the same posets.

\begin{theorem}
	\label{thm:DESintoSES}
	For each DES $ \delta $ there is a SES $ \sigma  = \ses{\delta} $, such that $  \posetsEq{\delta} {\sigma} $.
\end{theorem}

\begin{proof}[Proof of Theorem~\ref{thm:DESintoSES}]
	Let $ \delta = \left( E, \confOp, \buEnOp \right) $ be a DES. By \lem\ref{lem:DEStoSES}, $ \sigma = \ses{\delta} = \left( E, \confOp, \enabOp, \shrinkingCausality \right) $ is a SES such that $ \traces{\delta} = \traces{\sigma} $ and $ \configurations{\delta} = \configurations{\sigma} $.\\
	Let $ t = e_1 \cdots e_n \in \traces{\delta} $, $ 1 \leq i \leq n $, and the bundles $ \buEn{X_1}{e_i}, \ldots, \buEn{X_m}{e_i} $ all bundles pointing to $ e_i $.
	For $ U $ to be a cause for $ e_i $ \defi13 in \cite{dynamicCausality15} requires $ ( \ic{e_i} \setminus \dc{U}{e_i} ) \subseteq U $.
	Since $ \dc{H}{e} = \Set{ e' \mid \exists d \in H \logdot \drops{e'}{d}{e} } $ and $ \ic{e} = \{ e' \mid \enab{e'}{e} \} $, this condition holds iff the condition $ \enab{e'}{e_i} \implies e' \in U \lor \exists d \in U \logdot \drops{e'}{d}{e_i} $ holds for all $ e' \in E $.
	By \defi\ref{def:DESintoSES}, then $ ( \forall 1 \leq k \leq n \logdot X_k \cap U \neq \emptyset) $ iff $ \left( \left( \ic{e_i} \setminus \dc{U}{e_i} \right) \subseteq U \right) $.
	So, by \defs8 and 13 in \cite{dynamicCausality15}, $  \posetsEq{\delta}{\sigma} $.
\end{proof}

Thus SESs and DESs have the same expressive power.

\begin{proof}[Proof of Theorem~1 in \cite{dynamicCausality15}]
	By \ths\ref{thm:SESintoDES} and \ref{thm:DESintoSES}.
\end{proof}

\cite{Langerak97causalambiguity} proves that for DESs equivalence \wrt posets
based on early causality coincides with trace equivalence. Since SESs are as expressive as DESs \wrt families of posets based on early causality, the same correspondence holds for SESs.

\begin{corollary}
	\label{col:SESequiv}
	Let $ \sigma_1, \sigma_2 $ be two SES. Then $ \posetsEq{\sigma_1 }{\sigma_2} $ iff $ \traces{\sigma_1} = \traces{\sigma_2} $.
\end{corollary}

Then Theorem~2 in \cite{dynamicCausality15} states:
\begin{quote}
	Let $ \sigma, \sigma' $ be two SES.\\
	Then $ \posetsEq{\sigma }{ \sigma'} $ iff $ \transEq{\sigma}{\sigma'}$ iff $ \traces{\sigma} = \traces{\sigma'} $.
\end{quote}

\begin{proof}[Proof of Theorem~2 in \cite{dynamicCausality15}]
	By \cor\ref{col:SESequiv}, $  \posetsEq {\sigma}{\sigma' }$ iff $ \traces{\sigma} = \traces{\sigma'} $.
	
	If $ \configurations{\sigma} \neq \configurations{\sigma'} $ then, by \lem\ref{lem:SESconf} and \defs13 in \cite{dynamicCausality15}, $ \posetsNEq{\sigma } {\sigma'} $ and $\transNEq{\sigma}{\sigma'} $. Hence assume $ \configurations{\sigma} = \configurations{\sigma'} $.
	Note that, by Def. 13 and \lem\ref{lem:SESconf}, for all $ C \in \configurations{\sigma} $ there is a trace $ t \in \traces{\sigma} $ such that $ \overline{t} = C $. Moreover for every trace $ t \in \traces{\sigma} $ except the empty trace there is a sub-trace $ t' \in \traces{\sigma} $ and a sequence of events $ e_1 \cdots e_m $ such that $ t = t' e_1 \cdots e_m $ and $ \forall e \in \Set{ e_1, \ldots, e_m } \logdot \left( \ic{e} \setminus \dc{\overline{t'}}{e} \right) \subseteq \overline{t'} $.
	Thus, by \lem\ref{lem:SEStransTraces}, $ \traces{\sigma} = \traces{\sigma'} $ iff $  \transEq {\sigma}{\sigma'}$.
\end{proof}

Theorem~3 in \cite{dynamicCausality15} states:
\begin{quote}
	SESs and EBESs are incomparable.
\end{quote}

\begin{proof}[Proof of Theorem~3 in \cite{dynamicCausality15}]
	$ $\\
	Let $ \sigma_{\xi} = \left( \Set{ a, b, c }, \emptyset, \Set{ \enab{a}{b} }, \Set{ \drops{a}{c}{b} } \right) $ be the SES that is depicted in \fig\ref{fig:counterExamples}.
	Assume there is some EBES $ \xi = \left( E, \disaOp, \buEnOp \right) $ such that $ \traces{\sigma_\xi} = \traces{\xi} $.
	By Def. 13 in \cite{dynamicCausality15}, $ \traces{\sigma_\xi} = \Set{ \epsilon, a, c, ab, ac, ca, cb, abc, acb, cab, cba } $, \ie $ b $ cannot occur first.
	By \defi6 in \cite{dynamicCausality15}, a disabling $ \disa{x}{y} $ implies that $ y $ can never precedes $ x $.
	Thus we have $ \disaOp \cap \Set{ a, b, c }^2 = \emptyset $, because within $ \traces{\sigma_\xi} $ each pair of events of $ \Set{ a, b, c } $ occur in any order.
	Similarly we have $ \buEnOp \cap \{ \buEn{X}{e} \mid e \in \Set{ a, b, c } \land X \cap \Set{ a, b, c } = \emptyset \} = \emptyset $, because $ \buEn{x}{y} $ implies that $ x $ always has to precede $ y $.
	Moreover, by \defi6, adding impossible events as causes
	or using them within the disabling relation does not influence the set of
	traces.
	Thus there is no EBES $ \xi $ with the same traces as $ \sigma_{\xi} $. By
	\defi6 and the definition of posets in EBESs, then there is no EBES $ \xi $ with the same configurations or posets as $ \sigma_{\xi} $.
	
	Let $ \xi_{\sigma} = \left( \Set{ e, f }, \Set{ \disa{e}{f} }, \emptyset \right) $ be the EBES that is depicted in \fig\ref{fig:counterExamples}.
	Assume there is some SES $ \sigma = \left( E, \confOp, \enabOp, \shrinkingCausality \right) $ such that $ \traces{\xi_{\sigma}} = \traces{\sigma} $.
	According to \S~2.2 in \cite{dynamicCausality15}, $ \traces{\xi_{\sigma}} = \Set{ \epsilon,
	e, f, ef } $.
	By Def. 13 and because of the traces $ e $ and $ f $, there are no initial causes for $ e $ and f, \ie $ \enabOp \cap \Set{ \enab{x}{y} \mid y \in \Set{ e, f } } = \emptyset $.
	Moreover, $ \confOp \cap \Set{ e, f }^2 = \emptyset $, because of the trace $ ef $ and because conflicts cannot be dropped.
	Thus $ fe \in \traces{\sigma} $ but $ fe \notin \traces{\xi_{\sigma}} $, \ie
	there is no SES $ \sigma $ with the same traces as $ \xi_{\sigma} $. Then by
	\defi13, there is no SES $ \sigma $ with the same configurations or families of posets as $ \xi_{\sigma} $.
\end{proof}

\begin{lemma}\label{lma:SESinRCES}
	For each SES $ \sigma $ there is a RCES $ \rho $, such that $  \transEq{\sigma}{\rho}$.
\end{lemma}

\begin{proof}
	By \defi13 in \cite{dynamicCausality15}, $ \transS{X}{Y} $ implies $ X \subseteq Y $  for all $ X, Y \in \configurations{\sigma} $.
	
	Assume $ X \subseteq X' \subseteq Y' \subseteq Y $.
	Then, by Def. 13, $ \transS{X}{Y} $ implies $ \forall e, e' \in Y \logdot \neg \left( \conf{e}{e'} \right) $ and $ \forall e \in Y \setminus X \logdot \left( \ic{e} \setminus \dc{X}{e} \right) \subseteq X $.
	Then $ X \subseteq X' $ implies $ \left( \ic{e} \setminus \dc{X'}{e} \right)
	\subseteq \left( \ic{e} \setminus \dc{X}{e} \right) $. Then $ \forall e, e' \in Y' \logdot \neg \left( \conf{e}{e'} \right) $ and $ \forall e \in Y' \setminus X' \logdot \left( \ic{e} \setminus \dc{X'}{e} \right) \subseteq X' $, because of $ Y' \subseteq Y $.
	By Def. 13, then $ \transS{X'}{Y'} $.
	
	Thus $ \sigma $ satisfies the conditions of
	\defi\ref{def:translationIntoRCES}. Then by \lem\ref{lma:TransInRCES},
	$ \rho = \rces{\sigma} $ is a RCES such that $\transEq{\sigma}{\rho}$.
\end{proof}

\begin{lemma}
\label{lma:SESinRCESstrictly} 
	There is no transition-equivalent SES to the RCES $ \rho_{\sigma} $, where $ \rho_{\sigma} = ( \Set{ e , f }, \Set{ \enrc{\emptyset}{\Set{ e }} , \enrc{\emptyset}{\Set{ f }} , \enrc{\Set{ f }}{\Set{ e, f }} } ) $.
\end{lemma}

\begin{proof}
	Assume a SES $ \sigma = \left( E, \confOp, \enabOp, \shrinkingCausality \right) $ such that $ \transEq{\sigma }{\rho_{\sigma}} $. Then $ \configurations{\sigma} = \configurations{\rho_{\sigma}} $.
	By Def. 13 in \cite{dynamicCausality15} and \lem\ref{lem:SESconf} and because of
	the configuration $ \Set{ e, f } \in \configurations{\rho_{\sigma}} $, the
	events $ e $ and $ f $ cannot be in conflict with each other, \ie $ \confOp
	\cap \Set{ e, f }^2 = \emptyset $.
	Moreover, because of the configurations $ \Set{ e }, \Set{ f } \in \configurations{\rho_{\sigma}} $, there are no initial causes for $ e $ and $ f $, \ie $ \enabOp \cap \Set{ \enab{x}{y} \mid y \in \Set{ e, f } } = \emptyset $.
	Note that the relation $ \shrinkingCausality $ cannot disable events.
	Thus we have $ \forall a, b \in \Set{ e, f } \logdot \neg \left( \conf{a}{b} \right) $ and $ \left( \ic{e} \setminus \dc{\Set{ f }}{e} \right) = \emptyset \subseteq \Set{ f } $.
	But then, by Def. 13, $ \transS{\Set{ f }}{\Set{ e, f }} $.
	Since $ \transRC{\Set{ f }}{\Set{ e, f }} $ does not hold, this violates our
	assumption, \ie there is no SES which is transition equivalent to $
	\rho_{\sigma} $.
\end{proof}

Theorem~4 in \cite{dynamicCausality15} states:
\begin{quote}
	SESs are strictly less expressive than RCESs.
\end{quote}

\begin{proof}[Proof of Theorem~4 in \cite{dynamicCausality15}]
	By \lems\ref{lma:SESinRCESstrictly} and \ref{lma:SESinRCES}.
\end{proof}

\section{Alternative Partial Order Semantics in DES and SES}
\label{app:partialOrderSemantics}

To show that DES and SES are not only behavioral equivalent ES models but are also very closely related at the structural level we consider the remaining four intentional partial order semantics for DES of \cite{Langerak97causalambiguity}.

Liberal causality is the least restrictive notion of causality in \cite{Langerak97causalambiguity}. Here each set of events from bundles pointing to an event $ e $ that satisfies all bundles pointing to $ e $ is a cause.

\begin{definition}[Liberal Causality]
	Let $ \delta = \left( E, \confOp, \buEnOp \right) $ be a DES, $ e_1 \cdots e_n $ one of its traces, $ 1 \leq i \leq n $, and $ \buEn{X_1}{e_i}, \ldots, \buEn{X_m}{e_i} $ all bundles pointing to $ e_i $.
	A set $ U $ is a cause of $ e_i $ in $ e_1 \cdots e_n $ if
	\begin{compactitem}
		\item $ \forall e \in U \logdot \exists 1 \leq j < i \logdot e = e_j $,
		\item $ U \subseteq \left( X_1 \cup \ldots \cup X_m \right) $, and
		\item $ \forall 1 \leq k \leq m \logdot X_k \cap U \neq \emptyset $.
	\end{compactitem}
	Let $ \posetsLib{t} $ be the set of posets obtained this way for a trace $ t $.
\end{definition}

Bundle satisfaction causality is based on the idea that for an event $ e $ in a trace each bundle pointing to $ e $ is satisfies by exactly one event in a cause of $ e $.

\begin{definition}[Bundle Satisfaction Causality]
	Let $ \delta = \left( E, \confOp, \buEnOp \right) $ be a DES, $ e_1 \cdots e_n $ one of its traces, $ 1 \leq i \leq n $, and $ \buEn{X_1}{e_i}, \ldots, \buEn{X_m}{e_i} $ all bundles pointing to $ e_i $.
	A set $ U $ is a cause of $ e_i $ in $ e_1 \cdots e_n $ if
	\begin{compactitem}
		\item $ \forall e \in U \logdot \exists 1 \leq j < i \logdot e = e_j $ and
		\item there is a surjective mapping $ f : \Set{ X_k } \to U $ such that $ f\!\left( X_k \right) \in X_k $ for all $ 1 \leq k \leq m $.
	\end{compactitem}
	Let $ \posetsBsat{t} $ be the set of posets obtained this way for a trace $ t $.
\end{definition}

Minimal causality requires that there is no subset which is also a cause.

\begin{definition}[Minimal Causality]
	Let $ \delta = \left( E, \confOp, \buEnOp \right) $ be a DES and let $ e_1 \cdots e_n $ be  one of its traces, $ 1 \leq i \leq n $, and $ \buEn{X_1}{e_i}, \ldots, \buEn{X_m}{e_i} $ all bundles pointing to $ e_i $.
	A set $ U $ is a cause of $ e_i $ in $ e_1 \cdots e_n $ if
	\begin{compactitem}
		\item $ \forall e \in U \logdot \exists 1 \leq j < i \logdot e = e_j $,
		\item $ \forall 1 \leq k \leq m \logdot X_k \cap U \neq \emptyset $, and
		\item there is no proper subset of $ U $ satisfying the previous two conditions.	
	\end{compactitem}
	Let $ \posetsMin{t} $ be the set of posets obtained this way for a trace $ t $.
\end{definition}

Late causality contains the latest causes of an event that form a minimal set.

\begin{definition}[Late Causality]
	Let $ \delta = \left( E, \confOp, \buEnOp \right) $ be a DES, $ e_1 \cdots e_n $ one of its traces, $ 1 \leq i \leq n $, and $ \buEn{X_1}{e_i}, \ldots, \buEn{X_m}{e_i} $ all bundles pointing to $ e_i $.
	A set $ U $ is a cause of $ e_i $ in $ e_1 \cdots e_n $ if
	\begin{compactitem}
		\item $ \forall e \in U \logdot \exists 1 \leq j < i \logdot e = e_j $,
		\item $ \forall 1 \leq k \leq m \logdot X_k \cap U \neq \emptyset $,
		\item there is no proper subset of $ U $ satisfying the previous two conditions, and
		\item $ U $ is the latest set satisfying the previous three conditions.	
	\end{compactitem}
	Let $ \posetsLat{t} $ be the set of posets obtained this way for a trace $ t $.
\end{definition}

As derived in \cite{Langerak97causalambiguity}, it holds that
\begin{align*}
	\posetsLat{t}, \posetsEar{t} \subseteq \posetsMin{t} \subseteq \posetsBsat{t}
	\subseteq \posetsLib{t}
\end{align*}
for all traces $ t $.
Moreover a behavioral partial order semantics is defined and it is shown that two DESs have the same posets \wrt to the behavioral partial order semantics iff they have the same posets \wrt to the early partial order semantics iff they have the same traces.

Bundle satisfaction causality is---as the name suggests---closely related to the
existence of bundles. In SESs there are no bundles. Of course, as shown by the translation $ \des{\cdot} $ in \defi\ref{def:SESintoDES}, we can transform the initial and dropped causes of an event into a bundle. And of course if we do so an SES $ \sigma $ and its translation $ \des{\sigma} $ have exactly the same families of posets. But, because bundles are no native concept of SESs, we cannot directly map the definition of posets \wrt to bundle satisfaction to SESs.

To adapt the definitions of posets in the other three cases we have to replace
the condition $ U \subseteq \left( X_1 \cup \ldots \cup X_m \right) $ by $ U \subseteq \left( \Set{ e \mid \enab{e}{e_i} \lor \exists e' \in E \logdot \drops{e'}{e}{e_i} } \right) $ and replace the condition $ \forall 1 \leq k \leq m \logdot X_k \cap U \neq \emptyset $ by $ \left( \ic{e_i} \setminus \dc{U}{e_i} \right) \subseteq U $ (as in \defi13 in \cite{dynamicCausality15}). The remaining conditions remain the same with respect to traces as defined in Def. 13.
Let $ \posetsLib{t} $, $ \posetsMin{t} $, and $ \posetsLat{t} $ denote the sets of posets obtained this way for a trace $ t \in \traces{\sigma} $ of a SES $ \sigma $ \wrt liberal, minimal, and late causality. Moreover, let $ \posetsX{x}{\delta} = \bigcup_{t \in \traces{\delta}}{\posetsX{x}{t}} $ and $ \posetsX{x}{\sigma} = \bigcup_{t \in \traces{\sigma}}{\posetsX{x}{t}} $ for all $ x \in \Set{ \operatorname{lib}, \operatorname{bsat}, \operatorname{min}, \operatorname{late} } $.

Since again the definitions of posets in DESs and SESs are very similar the translations $ \des{\cdot} $ and $ \ses{\cdot} $ preserve families of posets. The proof is very similar to the proofs of \ths\ref{thm:SESintoDES} and \ref{thm:DESintoSES}.

\begin{theorem}
	For each SES $ \sigma $ there is a DES $ \delta $, namely $ \delta = \des{\sigma} $, and for each DES $ \delta $ there is a SES $ \sigma $, namely $ \sigma = \ses{\delta} $, such that $ \posetsX{x}{\sigma} = \posetsX{x}{\delta} $ for all $ x \in \Set{ \operatorname{lib}, \operatorname{min}, \operatorname{late} } $.
\end{theorem}

\begin{proof}
	The definitions of posets in DESs and SESs \wrt to minimal and late causality
	differ in exactly the same condition and its replacement as the definitions of
	posets in DESs and SESs \wrt early causality. Thus the proof in these two cases is similar to the proofs of \ths\ref{thm:SESintoDES} and \ref{thm:DESintoSES}.
	
	If $ \sigma = \left( E, \confOp, \enabOp, \shrinkingCausality \right) $ is a SES then, by \lem\ref{lem:SEStoDES}, $ \delta = \des{\sigma} = \left( E, \confOp, \buEnOp \right) $ is a DES such that $ \traces{\sigma} = \traces{\delta} $ and $ \configurations{\sigma} = \configurations{\delta} $.
	If $ \delta = \left( E, \confOp, \buEnOp \right) $ is a DES then, by \lem\ref{lem:DEStoSES}, $ \sigma = \ses{\delta} = \left( E, \confOp, \enabOp, \shrinkingCausality \right) $ is a DES such that $ \traces{\delta} = \traces{\sigma} $ and $ \configurations{\delta} = \configurations{\sigma} $.
	In both cases let $ t = e_1 \cdots e_n \in \traces{\sigma} $, $ 1 \leq i \leq n $, and $ \buEn{X_1}{e_i}, \ldots, \buEn{X_m}{e_i} $ be all bundles pointing to $ e_i $.
	
	In the case of liberal causality, for $ U $ to be a cause for $ e_i $ the
	definition of posets in SESs requires $ U \subseteq \left( \Set{ e \mid \enab{e}{e_i} \lor \exists e' \in E \logdot \drops{e'}{e}{e_i} } \right) $ and $ ( \ic{e_i} \setminus \dc{U}{e_i} ) \subseteq U $.
	The second condition holds iff $ \forall 1 \leq k \leq m \logdot X_k \cap U \neq \emptyset $ as shown in the proofs of \ths\ref{thm:SESintoDES} and \ref{thm:DESintoSES}.
	By \defs\ref{def:SESintoDES} and \ref{def:DESintoSES}, the first conditions holds iff $ U \subseteq \left( X_1 \cup \ldots \cup X_m \right) $.
	So, by the definitions of posets in DESs and SESs \wrt to liberal causality, $
	\posetsLib{\sigma} = \posetsLib{\delta} $.
\end{proof}

\section{Growing Causality}
As in SESs, both notions of configurations of GESs, traced-based and
transition-based; coincide and in different situations, the more suitable one
can be used.
\begin{lemma}\label{lma:GESConfigEquivalence}
	Let $ \gamma $ be a GES. Then $ \configTraces{\gamma} = \configurations{\gamma} $.
\end{lemma}
\begin{proof}
	Let $ \gamma = \left( E, \confOp, \enabOp, \growingCausality \right) $.
	
	By \defi15 in \cite{dynamicCausality15}, $ C \in \configTraces{\gamma} $ implies that there is some $ t = e_1 \cdots e_n $ such that $ \overline{t} \subseteq E $, $ \forall 1 \leq i, j \leq n \logdot \neg \left( \conf{e_i}{e_j} \right) $, $ \forall 1 \leq i \leq n \logdot \left( \ic{e_i} \cup \ac{\overline{t_{i - 1}}}{e_i} \right) \subseteq \overline{t_{i - 1}} $, and $ C = \overline{t} $. Hence, by \defi15, $ \transG{\overline{t_i}}{\overline{t_{i + 1}}} $ for all $ 1 \leq i \leq n $ and $ \transG{\emptyset}{\Set{ e_1 }} $.	Thus, by \defi15, $ C \in \configurations{\gamma} $.
	
	By \defi15, $ C \in \configurations{\gamma} $ implies that there are $ X_1, \ldots, X_n \subseteq E $ such that $ \transG{\transG{\transG{\emptyset}{X_1}}{\ldots}}{X_n} $ and $ X_n = C $.
	Then, by \defi15, we have:
	\begin{align}
		& \emptyset \subseteq X_1 \subseteq X_2 \subseteq \ldots \subseteq X_n \subseteq E & \tag{D1} \label{eq:D1}\\
		& \forall e, e' \in X_n \logdot \neg \left( \conf{e}{e'} \right) & \tag{D2} \label{eq:D2}\\
		& \forall e \in X_1 \logdot \left( \ic{e} \cup \ac{\emptyset}{e} \right) \subseteq \emptyset & \tag{D3} \label{eq:D3}\\
		& \begin{array}{l} \forall 1 \leq i < n \logdot \forall e \in X_{i + 1} \setminus X_i \logdot\\ \hspace{3em} \left( \ic{e} \cup \ac{X_i}{e} \right) \subseteq X_i \end{array} & \tag{D4} \label{eq:D4}\\
		& \begin{array}{l} \forall 1 \leq i < n \logdot \forall t, m \in X_{i + 1} \setminus X_i \logdot\forall c\in E \logdot\\ \hspace{3em} \addcause{c}{m}{t} \implies c\in  X_i  \end{array} & \tag{D5} \label{eq:D5}
	\end{align}
	Let $ X_1 = \Set{ e_{1, 1}, \ldots, e_{1, m_1} } $ and $ X_i \setminus X_{i - 1} = \Set{ e_{i, 1}, \ldots, e_{i, m_i} } $ for all $ 1 < i \leq n $.	Then, by \defi15, $ t = e_{1, 1} \cdots e_{1, m_1} \cdots e_{n, 1} \cdots e_{n, m_n} = e_1' \cdots e_k' $ is a trace such that $ \overline{t} \subseteq E $ (because of \eqref{eq:D1}), $ \neg \left( \conf{e_i'}{e_j'} \right) $ for all $ 1 \leq i, j \leq k $ (because of \eqref{eq:D1} and \eqref{eq:D2}), for all $ 1 \leq i \leq k $ and all $ 1 \leq j \leq m_i $ we have $ \left( \ic{e_{i, j}} \cup \ac{\overline{t_{i - 1}}}{e_{i, j}} \right) \subseteq \overline{t_{i - 1}} $ (because of \eqref{eq:D3}, \eqref{eq:D4}, and, by \eqref{eq:D5}, $ \ac{\overline{t_{i - 1}} \cup X_i}{e_{i, j}} = \ac{\overline{t_{i - 1}}}{e_{i, j}} $), and $ \overline{t} = C $ (because $ X_n = C $).
	Thus $ C \in \configTraces{\gamma} $.
\end{proof}

For the incomparability result between GESs and EBESs we consider two
counterexamples, and show that there is no equivalent EBES or GES respectively.
\begin{lemma}\label{lma:EBESninGES}
	There is no configuration-equivalent GES to $\beta_\gamma$ (\cf \fig\ref{fig:counterExamples}).
\end{lemma}

\begin{proof}
	Assume a GES $ \gamma = \left( E, \confOp', \enabOp, \growingCausality \right) $ such that $ \configurations{\gamma} = \configurations{\beta_\gamma} $.	According to \S~2.2 in \cite{dynamicCausality15}, $ \configurations{\beta_\gamma} = \Set{\emptyset, \Set{ a }, \Set{ b }, \Set{ a, c }, \Set{ b, c } } $. Because $ \Set{ c } \notin \configurations{\beta_\gamma} $, $ \Set{ a, c } \in \configurations{\beta_\gamma} $, and by \defi15 in \cite{dynamicCausality15} and \lem\ref{lma:GESConfigEquivalence}, $ a $ has to be an initial cause of $ c $ in $ \gamma $, \ie $ \enab{a}{c} $.	But then, by \defi15 and \lem\ref{lma:GESConfigEquivalence}, $ \Set{ b, c } \notin \configurations{\gamma} $ although $ \Set{ b, c } \in \configurations{\beta_\gamma} $. This violates our assumption, \ie no GES can be configuration equivalent to $ \beta_\gamma $.
\end{proof}
\begin{lemma}\label{lma:GESninEBES}
	There is no trace-equivalent EBES to $\gamma_\xi$ (\cf \fig\ref{fig:counterExamples}).
\end{lemma}
\begin{proof}
	Assume a EBES $ \xi = \left( E, \confOp, \buEnOp \right) $ such that $ \traces{\xi} = \traces{\gamma_\xi} $. By \defi15 in \cite{dynamicCausality15}, $ a, c, ca, bac \in \traces{\gamma_\xi} $ and $ ac \notin \traces{\gamma_\xi} $. Because of $ a, c \in \traces{\gamma_\xi} $ and by \defi6 in \cite{dynamicCausality15}, $ a $ and $ c $ have to be initially enabled in $ \xi $, \ie $ \buEnOp \cap \Set{ \buEn{X}{y} \mid y \in \Set{ a, c } } = \emptyset $.	Moreover, because of $ ca, bac \in \traces{\gamma_\xi} $, $ a $ cannot disable $ c $, \ie $ \neg \left( \disa{a}{c} \right) $. But then $ ac \in \traces{\xi} $. This violates our assumption, \ie there is no trace-equivalent EBES to $ \gamma_\xi $.
\end{proof}
Theorem~5 in \cite{dynamicCausality15} states:
\begin{quote}
	GESs are incomparable to BESs and EBESs.
\end{quote}

\begin{proof}[Proof of Theorem~5 in \cite{dynamicCausality15}]
	By \lem\ref{lma:EBESninGES}, there is no GES that is configuration equivalent
	to the BES $ \beta_\gamma $. Thus no GES can have the same families of posets
	as the BES $ \beta_\gamma $, because two BES with different configurations
	cannot have the same families of posets (\cf \S~2.2 in \cite{dynamicCausality15}). Moreover, by \defs3 and 5 in \cite{dynamicCausality15}, each BES is also an EBES. Thus no GES can have the same families of posets as the EBES $ \beta_\gamma $.
	
	By \lem\ref{lma:GESninEBES}, there is no EBES and thus also no BES that is trace-equivalent to the GES $ \gamma_\xi $. By \defi15 in \cite{dynamicCausality15}, two GES with different traces cannot have the same transition graphs. Thus no EBES or BES can be transition-equivalent to $ \gamma_\xi $.
\end{proof}

For the incomparability between GESs and SESs, we study a GES counterexample, such that no SES is trace-equivalent.
\begin{lemma}\label{lma:GESninSES}
There is no trace-equivalent SES to $\gamma_\sigma$ (\cf \fig\ref{fig:counterExamples}).
\end{lemma}
\begin{proof}
	Assume a SES $ \sigma = \left( E, \confOp, \enabOp, \shrinkingCausality \right) $ such that $ \traces{\sigma} = \traces{\gamma_\sigma} $.	By \defi15 in \cite{dynamicCausality15}, $ \traces{\gamma_\sigma} = \Set{ \epsilon, a, b, ab } $.	Because of the trace $ ab \in \traces{\gamma_\sigma} $ and by Def. 13 in \cite{dynamicCausality15}, $ a $ and $ b $ cannot be in conflict, \ie $ \neg (\conf{a}{b} ) $ and $ \neg ( \conf{b}{a} ) $.	Moreover, because of the traces $ a, b \in \traces{\gamma_\sigma} $, there are no initial cases for $ a $ or $ b $, \ie $ \enabOp \cap \Set{ \enab{x}{y} \mid y \in \Set{ a, b } } = \emptyset $.	Thus, by Def. 13, $ ba \in \traces{\sigma} $ but $ ba \notin \traces{\gamma_\sigma} $.	This violates our assumption, \ie no SES can be trace equivalent to $	\gamma_\sigma $.
\end{proof}
Theorem~6 in \cite{dynamicCausality15} states:
\begin{quote}
	GESs and SESs are incomparable.
\end{quote}
\begin{proof}[Theorem~6 in \cite{dynamicCausality15}]
	By \lem\ref{lma:GESninSES}, no SES is trace-equivalent to the GES $ \gamma_\sigma $. By \defi15 in \cite{dynamicCausality15}, two GES with different traces cannot have the same transition graphs. Thus no SES is transition-equivalent to the GES $ \gamma_\sigma $.
	
	By \cite{Langerak:Thesis}, BESs are less expressive than EBESs and by \cite{Langerak97causalambiguity}, BESs are less expressive than DESs.	By \theo5 in \cite{dynamicCausality15}, BESs and GESs are incomparable an by \theo1 in \cite{dynamicCausality15} DESs are as expressive as SESs.	Thus GESs and SESs are incomparable.
\end{proof}
To show that GESs are strictly less expressive than RCESs, we give a translation
for one direction and a counterexample for the other.
\begin{lemma}\label{lma:GESinRCES}
	For each GES $ \gamma $ there is an RCES $ \rho $, such that $ \transEq{\gamma }{ \rho} $.
\end{lemma}
\begin{proof}
	Let $ \gamma = \left( E, \confOp, \enabOp, \growingCausality \right) $.	By \defi15 in \cite{dynamicCausality15}, $ \transG{X}{Y} $ implies $ X \subseteq Y $.
	
	Assume $ X \subseteq X' \subseteq Y' \subseteq Y $ and $ \transG{X}{Y} $.
	By \defi15, then we have $ \forall e, e' \in Y' \logdot \neg \left( \conf{e}{e'} \right) $,	$ \forall e \in \left( Y' \setminus X' \right) \logdot \left( \ic{e} \cup \ac{X}{e} \right) \subseteq X $, and $ \forall t, m \in Y\setminus X\logdot\forall c\in E \logdot \addcause{c}{m}{t} \implies c\in X $.	Moreover, because $ \forall t, m \in Y\setminus X\logdot\forall c\in E \logdot \addcause{c}{m}{t} \implies c\in X $, $ \ac{X}{e} = \ac{X'}{e} $ for all $ e \in Y' \setminus X' $.	Hence $ \forall e \in \left( Y' \setminus X' \right) \logdot \left( \ic{e} \cup \ac{X'}{e} \right) \subseteq X' $ and $ \forall t, m \in Y'\setminus X'\logdot\forall c\in E \logdot \addcause{c}{m}{t} \implies c\in X' $.	Thus, by \defi15, $ \transG{X'}{Y'} $.
	
	By \lem\ref{lma:TransInRCES}, $\rho=\rces{\gamma} $ is an RCES and $ \transEq{\gamma}{\rho}$.
\end{proof}
\begin{lemma}\label{lma:GESinRCESstrictly}
	There is no transition-equivalent GES to $\rho_{\gamma}$ (\cf \fig3 in \cite{dynamicCausality15}).
\end{lemma}
\begin{proof}
	Assume a GES $ \gamma = \left( E, \confOp, \enabOp, \growingCausality \right) $ such that $ \transEq {\gamma}{\rho_{\gamma}} $. Then $ \configurations{\gamma} = \configurations{\rho_{\gamma}} $.	By \defi15 in \cite{dynamicCausality15} and because of the configuration $ \Set{ a, b, c } \in \configurations{\rho_{\gamma}} $, the events $ a $, $ b $, and $ c $ cannot be in conflict with each other, \ie $ \confOp \cap \Set{ a,	b, c }^2 = \emptyset $.	Moreover, because of the configurations $ \Set{ a }, \Set{ b }, \Set{ c } \in \configurations{\rho_{\gamma}} $, there are no initial causes for $ a $, $ b $, or $ c $, \ie $ \enabOp \cap \Set{ \enab{x}{y} \mid y \in \Set{ a, b, c } } = \emptyset $. Finally, because of the configurations $ \Set{ a, c }, \Set{ b, c } \in \configurations{\rho_{\gamma}} $, neither $ a $ nor $ b $ can add a cause (except of themselves) to $ c $, \ie $ \addcause{e}{a}{c} \implies e = a $ and $ \addcause{e}{b}{c} \implies e = b $ for all $ e \in E $.	Thus we have $ \forall e, e' \in \Set{ a, b, c } \logdot \neg \left( \conf{e}{e'} \right) $ and $ \left( \ic{c} \cup \ac{\Set{ a, b }}{c} \right) = \emptyset \subseteq \Set{ a, b } $.	But then, by \defi15, $ \transG{\Set{ a, b }}{\Set{ a, b, c }} $. Since $ \neg \left( \transRC{\Set{ a, b }}{\Set{ a, b, c }} \right) $, this violates our assumption, \ie there is no GES that is transition equivalent to $ \rho_{\gamma} $.
\end{proof}
Theorem~7 in \cite{dynamicCausality15} states:
\begin{quote}
	GESs are strictly less expressive than RCESs.
\end{quote}
\begin{proof}[Proof of Theorem~7 in \cite{dynamicCausality15}]
	By \lems\ref{lma:GESinRCES} and \ref{lma:GESinRCESstrictly}.
\end{proof}
\section{Dynamic Causality}
In order to justify our approach of state transition equivalence, we need a notion of (configuration) transition equivalence, and show that the new equivalence is needed.
\begin{definition}
\label{def:DCESConfigs}
Let $\Delta$ be a DCES. The set of its (reachable) configurations is
$\configurations{\Delta}=\pi_1(\reachables{\Delta})$; the projection on the first component of the states.
\end{definition}
Lemma~1 in \cite{dynamicCausality15} states:
\begin{quote}
	There are DCESs that are transition equivalent but not state transition equivalent.
\end{quote}
\begin{proof}[Proof of Lemma~1 in \cite{dynamicCausality15}]
We consider two DCESs
	$\Delta=(\Set{a,b,c,d},\emptyset,\emptyset,\\ \Set{\drops{c}{b}{d}}, \Set{\addcause{c}{a}{d}})$ and
	$\Delta'=(\Set{a,b,c,d},\emptyset,\emptyset,\emptyset,\emptyset)$.
In $\Delta$ there is a transition $\transDC{(\emptyset,\csiD)}{(\Set{a},\csD')}$
by \defi18 in \cite{dynamicCausality15}. Initially the causality function is the
constant empty set function (\cf definition of $\csiD$) in
\defi17 in \cite{dynamicCausality15}). After $a$ occurs
$\cs[']{d}=\Set{c}$ is updated according to
\defi18 Condition~6. Next we have
$\transDC{(\Set{a},\csD')}{(\Set{a,b},\csD'')}$, where 
$\cs['']{d}=\emptyset$ according to
Condition~4. Now there is a possible transition to
$(\Set{a,b,d},\csD''')$. But if we
proceed from $(\emptyset,\csiD)$ to $(\Set{b},\csD')$ with
$\cs[']{d}=\emptyset$ according to
Condition~4 and then to $(\Set{a,b},\csD'')$ with
$\cs['']{d}=\Set{c}$ according to Condition~6 there is no transition to $(\{a,b,d\},\csD''')$, because $c$ needs $d$ according to Condition~3.
In  $\Delta'$ both sequences of state transitions are possible by \defi18. Thus $\Delta$ and $\Delta'$ are not state transition equivalent.

On the other hand, if we only consider the configurations of $\Delta$ and $\Delta'$ saying there is a transition from $C$ to $C'$ whenever $ \transDC{(C,\csD)}{(C',\csD')} $, then $\Delta$ and $\Delta'$ are transition equivalent.
\end{proof}

To compare DCESs to other ESs we define the Single State Dynamic Causality ESs
(SSDCs) as a subclass of DCESs.

\begin{definition}
\label{def:SingleStateDC}
Let SSDC be a subclass of DCESs such that $\varrho$ is a SSDC iff
$\forall e, e' \in E \logdot \nexists a,d \in E
		\logdot \addsDrops{a}{e'}{e}{d}$.
\end{definition}
Since there are no adders and droppers for the same causal dependency, the order
of modifiers does not matter and thus there are no two different states sharing
the same configuration, \ie each configuration represents a state. Thus it is
enough for SSDC to consider transition equivalence with respect to
configurations, \ie $\transEqOp$.

\begin{lemma}
\label{lma:SingleCausalState}
Let $\varrho $ be a SSDC. Then for the causal-state function $cs$ of any
state $(C, cs) \in \reachables{\varrho}$ it holds $cs(e) = (\ic{e} \cup
\ac{C}{e}) \setminus (\dc{C}{e}\cup C)$.
\end{lemma}

\begin{proof}
If $C=\emptyset$ the equation follows directly from the definitions of $\csD$, $\icD$, $\acD$, and $\dcD$.

Assume $\transDC{(C,\csD)}{(C',\csD')}$. By induction, we have $cs(e) = (\ic{e} \cup\ac{C}{e}) \setminus (\dc{C}{e}\cup C)$.
We prove for each $e\in E\setminus C'$ by a doubled case distinction $\cs[']{e}=(\ic{e} \cup\ac{C'}{e}) \setminus (\dc{C'}{e}\cup C)$. Let us first assume $e'\notin \cs{e}$ but $e'\in \cs[']{e}$, then by Condition~6 in \cite{dynamicCausality15} we have $\exists a\in C'\setminus C\logdot \addcause{e'}{a}{e}$ and since $\ac{C'}{e}=\Set{ e' \mid \exists a \in C'\logdot \addcause{e'}{a}{e} \land a \notin \Set{ e, e' } }$ we have $e'\in \ac{C'}{e}$, because $\varrho$ is a SSDC it follows $e'\notin\dc{C'}{e}$ and because $e\in E\setminus C'$ it follows $e\notin C$. Then in this case $e'\in(\ic{e} \cup\ac{C'}{e}) \setminus (\dc{C'}{e}\cup C)$ holds. Let now still $e'\notin \cs{e}$ but $e'\notin \cs[']{e}$, then we have by contra-position of Condition~7 we have $\nexists a\in C'\setminus C\logdot \addcause{e'}{a}{e}$, and so  $e'\notin(\ic{e} \cup\ac{C'}{e}) \setminus (\dc{C'}{e}\cup C)$. Let us now consider the case $e'\in\cs{e}$ and here first $e'\in\cs[']{e}$. Then by Condition~5 it follows $\nexists d\in C'\setminus C\logdot \drops{e'}{d}{e}$. Then $e'\notin\dc{C'}{e}$ and because $e\in E\setminus C'$ it follows $e\notin C$ and so $e'\notin(\ic{e} \cup\ac{C'}{e}) \setminus (\dc{C'}{e}\cup C)$. In the last case we consider $e'\in\cs{e}$ and  $e'\notin\cs[']{e}$. By Condition~4 we have $\exists d\in C'\setminus C\logdot \drops{e'}{d}{e}$ and so $e'\in\dc{C'}{e}$. Thus $e'\notin(\ic{e} \cup\ac{C'}{e}) \setminus (\dc{C'}{e}\cup C)$. So in each case $\cs[']{e}=(\ic{e} \cup\ac{C'}{e}) \setminus (\dc{C'}{e}\cup C)$ holds.
\end{proof}

In SSDC Conditions~4, 5, 6, and 7 in \cite{dynamicCausality15} hold whenever $ C \subseteq C' $.

\begin{lemma}
\label{lma:SSDCstateProp}
Let $\rho $ be a SSDC and let $(C,\csD)$ and $(C',\csD')$ be two states of $\rho$ with $C\subseteq C'$, then Conditions~4, 5, 6, and 7 in \cite{dynamicCausality15} of $\transDCOp$ hold for those two states.
\end{lemma}

\begin{proof}
Let $e'\in E\setminus C'$ with $e'\in\cs{e}\setminus\cs[']{e}$.\\
Since $\ac{C}{e}=\Set{ e' \mid \exists a \in C\logdot \addcause{e'}{a}{e} \land a \notin \Set{ e, e' } }$, $\dc{C}{e} = \{ e' \mid \exists d \in C \logdot \drops{e'}{d}{e} \}$ and $ C \subseteq C' $, we have $\ac{C}{e}\subseteq\ac{C'}{e}$ and $\dc{C}{e}\subseteq\dc{C'}{e}$ and by the previous \lem\ref{lma:SingleCausalState} we have $e'\in ((\ic{e} \cup\ac{C}{e}) \setminus (\dc{C}{e}\cup C))\setminus((\ic{e} \cup\ac{C'}{e}) \setminus (\dc{C'}{e}\cup C'))$, so $e'\in((\ic{e} \cup\ac{C}{e}) \setminus (\dc{C}{e}\cup C))$ and $e'\notin((\ic{e} \cup\ac{C'}{e}) \setminus (\dc{C'}{e}\cup C'))$. Then $e'\in\ic{e}\cup\ac{C}{e}$ and so $e'\in\ic{e}\cup\ac{C'}{e}$, thus $e'\in(\dc{C'}{e}\cup C')$ and so $e'\in\dc{C'}{e}$, but $e'\notin\dc{C}{e}$, which yields $\droppers{e'}{e}\cap(C'\setminus C)\neq\emptyset$, so Condition~4 in \cite{dynamicCausality15} holds. Let now $e,e'\in E\setminus C'$ with $\droppers{e'}{e}\cap(C'\setminus C)\neq\emptyset$, then $e'\in\dc{C'}{e}$ so $e'\notin\cs[']{e}$ follows, which is exactly 5. 

Conditions~6 and 7 are proven similarly.
\end{proof}

\begin{lemma}
\label{lma:SSDCConfInTrans}
Let $\rho$ be a SSDC and $\transDC{(X,\csD_X)}{(Y,\csD_Y)}$ a transition in $\rho$. Then for all $X',Y'$ with $X\subseteq X'\subseteq Y'\subseteq Y$, there is a transition $\transDC{(X',\csD_X')}{(Y',\csD_Y')}$ in $\rho$, where $\csD_{X'}(e)=((\ic{e} \cup\ac{X'}{e}) \setminus (\dc{X'}{e}\cup X'))$ and $\csD_{Y'}(e) = ((\ic{e} \cup\ac{Y'}{e}) \setminus (\dc{Y'}{e}\cup Y'))$.
\end{lemma}

\begin{proof}
By assumption Conditions~1 and 2 of \defi18 in \cite{dynamicCausality15} of the
transition relation holds for the two states $(X',\csD_{X'})$ and
$(Y',\csD_{Y'})$. Conditions~4,
5, 6, and
7 follow from \lem\ref{lma:SSDCstateProp}.
Condition~8 holds because of \defi\ref{def:SingleStateDC} and $\rho$ is a SSDC. Condition~9 holds because it is a special case of the same conditions for $(X,\csD_X)$ and $(Y,\csD_Y)$. Let now $e\in Y'\setminus X'$, such that $\cs{e}_{X'}\neq \emptyset$, then there is $a\in X'\setminus X$ and a $ c \in E$ with $\addcause{c}{a}{e}$, but this is a contradiction with Condition~9 of $\transDC{(X,\csD_X)}{(Y,\csD_Y)}$, so Condition~3 holds. Thus $\transDC{(X',\csD_X')}{(Y',\csD_Y')}$ holds.
\end{proof}

\begin{definition}
\label{def:InclusionIntoDCES}
Let $\sigma=(E,\confOp,\enabOp,\shrinkingCausality)$ be a SES. Then its embedding is $\emb{\sigma}=(E,\confOp,\enabOp,\shrinkingCausality,\emptyset)$. Similarly let $\gamma=(E,\confOp,\enabOp,\growingCausality)$ be a GES. Then its embedding is $\emb{\gamma}=(E,\confOp,\enabOp,\emptyset,\growingCausality)$.
\end{definition}

For each embedding the causal state coincides with a condition on the initial, added, and dropped causes, that are enforced in the transition relations of SESs and GESs.

\begin{lemma}
\label{lma:explicitSEScs}
Let $\sigma$ be a SES and $\emb{\sigma}$ its embedding. Then we have for each state $(C,\csD)$ of $\emb{\sigma}$, $\cs{e}=\ic{e}\setminus(\dc{C}{e}\cup C)$.
\end{lemma}

\begin{proof}
By \lem\ref{lma:SingleCausalState} and because $ \ac{C}{e} = \emptyset $ in $ \emb{\sigma} $ for all configurations $ C $ and events $ e $.
\end{proof}

\begin{lemma}
\label{lma:explicitGEScs}
Let $\gamma$ be a GES and $\emb{\gamma}$ its embedding. Then we have for each state $(C,\csD)$ of $\emb{\gamma}$,  $\cs{e}=(\ic{e}\cup\ac{C}{e})\setminus C$.
\end{lemma}

\begin{proof}
By \lem\ref{lma:SingleCausalState} and because $ \dc{C}{e} = \emptyset $ in $ \emb{\gamma} $ for all configurations $ C $ and events $ e $.
\end{proof}
SESs (resp. GESs) and their embeddings are transition equivalent.
\begin{lemma}
\label{lma:SESinDCES}
\label{lma:GESinDCES}
Let $\mu$ be a GES or SES, then we have $\transEq{\emb{\mu} }{\mu}$.
\end{lemma}

%
%Lemma~\ref{lma:SESinDCES} states:
%\begin{quote}
%For each SES $\sigma$, $\emb{\sigma} \transEq \sigma$.
%\end{quote}

\begin{proof}
Let $\mu$ be a SES and $\transS{C}{C'}$ a transition in $\mu$, we define for a configuration $C$ a causality state function $\csD:E\setminus X\rightarrow\Powerset{E\setminus X}$ as $\cs{e}=\ic{e}\setminus(\dc{X}{E}\cup X)$. 
$C'$ is conflict free and $C\subseteq C'$, because $\transS{C}{C'}$ and Def. 13 in \cite{dynamicCausality15}, so Conditions~
1 and 2 of \defi18 in \cite{dynamicCausality15} are satisfied. Moreover in the configuration $C$ we have $\cs{e}=\ic{e}\setminus(\dc{C}{E}\cup C)$, so Conditions~3, 4, and 5 are fulfilled. Conditions~6, 8, 7, and 9 are trivially satisfied, because $\growingCausality=\emptyset$, so $\transDC{(C,\csD)}{(C',\csD')}$.
Let now $\transDC{(C,\csD)}{(C',\csD')}$ in $\emb{\mu}$, then by \defs18, 13 in combination with \lem\ref{lma:explicitSEScs} there is a transition $\transS{C}{C'}$ in $\mu$.\\
Let now $\mu$ be a GES and $\transG{C}{C'}$ a transition in $\mu$, we define for a configuration $X$ a causality state function $\csD:E\setminus X\rightarrow\Powerset{E\setminus X}$ as $\cs{e}=(\ic{e}\cup\ac{X}{e})\setminus X$. $C'$ is conflict free and $C\subseteq C'$, because $\transS{C}{C'}$ and \defi15 in \cite{dynamicCausality15}, so Conditions~2, 1,  and 9 of \defi18 are satisfied. Moreover in the configuration $C$ we have $\cs{e}=(\ic{e}\cup\ac{C}{e})\setminus C$, so Conditions~3, 6, and 7 are fulfilled. Conditions~4, 5, and 8 are trivially satisfied, because $\shrinkingCausality=\emptyset$.
Let now $\transDC{(C,\csD)}{(C',\csD')}$ in $\emb{\mu}$, then by \defs18, 15 in combination with \lem\ref{lma:explicitGEScs} there is a transition $\transG{C}{C'}$ in $\mu$. 
\end{proof}

For the incomparability result between DCESs and SESs, we give an RCES
counterexample, which cannot be modeled by a DCES.

\begin{lemma}\label{lma:RCESNotinDCES}
There is no transition-equivalent DCES to $\rho_{\gamma}$ (\cf \fig\ref{fig:counterExamples}).
\end{lemma}

\begin{proof}[Proof of Lemma~\ref{lma:RCESNotinDCES}]
Assume $ \Delta = \left( E, \confOp, \enabOp,\shrinkingCausality, \growingCausality \right) $ such that $ \transEq{\Delta}{\rho_{\gamma}} $. Then $ \configurations{\Delta} = \configurations{\rho_{\gamma}} $.
	By \defi18 in \cite{dynamicCausality15} and because of the configuration $ \Set{
	a, b, c } \in \configurations{\rho_{\gamma}} $, the events $ a $, $ b $, and $
	c $ cannot be in conflict with each other, \ie $ \confOp \cap \Set{ a, b, c }^2
	= \emptyset $.
	Moreover, because of the configurations $ \Set{ a }, \Set{ b }, \Set{ c } \in \configurations{\rho_{\gamma}} $, there are no initial causes for $ a $, $ b $, or $ c $, \ie $ \enabOp \cap \Set{ \enab{x}{y} \mid y \in \Set{ a, b, c } } = \emptyset $. Note that the relation $\shrinkingCausality$ cannot disable events.
	Finally, because of the configurations $ \Set{ a, c }, \Set{ b, c } \in \configurations{\rho_{\gamma}} $, neither $ a $ nor $ b $ can add a cause (except of themselves) to $ c $, \ie $ \addcause{e}{a}{c} \implies e = a $ and $ \addcause{e}{b}{c} \implies e = b $ for all $ e \in E $.
	Thus we have $ \forall e, e' \in \Set{ a, b, c } \logdot \neg \left( \conf{e}{e'} \right) $ and in the state $(\Set{a,b}, cs)$ it follows $ \cs{c} = \emptyset \subseteq \Set{ a, b } $.
	But then, by \defi18, $ \transDC{(\Set{ a, b },\csD)}{(\Set{ a, b, c },\csD')} $ for some causal state functions $\csD$ and $\csD'$.
	Since $ \neg \left( \transRC{\Set{ a, b }}{\Set{ a, b, c }} \right) $, this violates our assumption, \ie there is no DCES that is transition equivalent to $ \rho_{\gamma} $.
\end{proof}

Theorem~8 in \cite{dynamicCausality15} states:
\begin{quote}
DCESs and RCESs are incomparable.
\end{quote}

\begin{proof}[Proof of Theorem~8 in \cite{dynamicCausality15}]
It follows from \lems\ref{lma:RCESNotinDCES} and 1 in \cite{dynamicCausality15}, and because $\transEqS{\Delta}{\Delta'}$ (for $\Delta$ and $\Delta'$ as in the proof of \lem1), then no two RCESs $\rho$ and $\rho'$, with $\transNEq{\rho}{\rho'}$ can distinguish between $\Delta$ and $\Delta'$.
\end{proof}

Theorem~9 in \cite{dynamicCausality15} states:
\begin{quote}
DCESs are strictly more expressive than GESs and SESs.
\end{quote}

\begin{proof}[Proof of Theorem~9 in \cite{dynamicCausality15}]
By \ths8, 7, and 4 in \cite{dynamicCausality15} and \lem\ref{lma:GESinDCES}.
\end{proof}

\section{Comparing DCESs with EBESs}
\label{sec:EBESsWithDCESs}
To compare with EBESs, we define a sub-class of DCESs, where posets could be
defined and used for semantics.

\begin{definition}
\label{def:EBDC}
Let EBDC denotes a subclass of SSDC with the additional requirements:
	\begin{inparaenum}
		\item \label{eq:EBDCOnlyDisabling} $\forall c,m,t \in E \logdot
		\addcause{c}{m}{t} \Longrightarrow c = t$
		\item \label{eq:NoCausalAmbiguity} $\forall c,m_1,\ldots , m_n, t \in E
		\logdot \drops{c}{m_1,\ldots,m_n}{t} \Longrightarrow \forall e_1, e_2 \in 
		\{c,m_1,\ldots , m_n\} \logdot \left(e_1 \neq e_2 \Longrightarrow
		e_1 \# e_2 \right)$
	\end{inparaenum}
\end{definition}

The first condition translates disabling into $ \growingCausality $ and ensures that
disabled events cannot be enabled again. The second condition reflects causal
unambiguity by $ \shrinkingCausality $ such that either the initial
cause or one of its droppers can happen.

We adapt the notion of precedence.

\begin{definition}
\label{def:EBDCLposets}
Let $\vartheta$ be a EBDC and $C \in C(\vartheta)$, then we define the
precedence relation $<_C \subseteq C\times C$ as $e <_C e'
\Longleftrightarrow e \rightarrow e' \lor \addcause{e}{e'}{e} \lor \exists c
\in E\logdot \drops{c}{e}{e'}$. Let $\leq_C$ be the reflexive and
transitive closure of $<_C$.
\end{definition}

The relation $<_C$ indeed represents a precedence relation, and its
reflexive transitive closure is a partial order.

\begin{lemma}
\label{lma:EBESPrecedence}
Let $\vartheta$ be a EBDC, $C \in C(\vartheta)$, and let $e,e' \in C \logdot e
<_C e'$. Let also $ \transDC{\transDC{(C_0, cs_0)}{\ldots}}{(C_n, cs_n)}$ with $ C_0 = \emptyset $ and $ C_n = C $ be the
transition sequence of $C$,
then $\exists C_i \in \{C_0,\ldots,C_n\} \logdot e \in C_i \land e' \notin C_i$.
\end{lemma}

\begin{proof}
Let $(C_f, cs_f)$ be the first occurrence of $e$ in the sequence $ \transDC{\transDC{(C_0, cs_0)}{\ldots}}{(C_n, cs_n)}$, so according to
Condition~1 of \defi18 in \cite{dynamicCausality15} it is enough to prove that
$e' \notin C_f$.
First, assume that $e \rightarrow e'$, then $e \in cs_0(e')$
according to the definition of $\csiD$. Then
according to \defi18 the only situation where $e \notin cs_{f-1}(e')$ is that
there is a dropper $e'' \in C_{f-1}$ for it according to
Condition~4, but that is impossible
since $e$ and $e'$ will be in conflict according to
Condition~\ref{eq:NoCausalAmbiguity} of \defi\ref{def:EBDC}. So $e \in cs_{f-1}(e')$ and thus $e'
\notin C_f$ according to Condition~3 of \defi18.

Second assume that $\addcause{e}{e'}{e}$.
If $e' \in C_{f-1}$ then according to
Condition~7 of \defi18 $e \in cs_{f-1}(e)$ which
means $e \notin C_f$ according to
Condition~3 of \defi18, which is a contradiction to the definition of $C_f$. 
Then according to Condition~9 of \defi18, if $e'
\in C_{f}$, it follows $e\in C_{f-1}$, which again contradicts the definition of $C_f$. So because $e'\in C$, there is an $h>f$, such that $e'\in C_h$ but $e'\notin C_{h-1}$.

Third, assume $\exists c
\in E\logdot \drops{c}{e}{e'}$. Then since EBDC are a subclass of SSDC we have
$\nexists a \in E \logdot \addcause{c}{a}{e'}$ according to \defi\ref{def:SingleStateDC}. Then $c \rightarrow e'$ according to
Condition~1 of \defi16 in \cite{dynamicCausality15}, which means $c \in
cs_0(e')$ according to definition of $\csiD$ in \defi18.
Let us assume that $e' \in C_f$ then either $c$ or another dropper $d \logdot
\drops{c}{d}{e'}$ occurred before $e'$, which is impossible because of the
mutual conflict in Condition~\ref{eq:NoCausalAmbiguity} of \defi\ref{def:EBDC}. So $e' \notin C_f$.
\end{proof}

\begin{lemma}
	\label{lma:EBESPrecIsOrder}
	$\leq_C$ is a partial order over $C$. 
\end{lemma}

\begin{proof}
Let $e,e' \in C \logdot e <_C e'$ and let $(\emptyset = C_0, cs_0) \ldots
(C_n = C, cs_n)$ be the transition sequence of $C$. Let also $C_h, C_j$ be the configurations
where $e, e'$ first occur, respectively, then according to \lem\ref{lma:EBESPrecedence}, $h<j$. Since $\leq_C$ is the reflexive and transitive
closure of $<_C$, then $e \leq_C e' \Longrightarrow h \leq j$. For
anti-symmetry, assume that $e' \leq_C e$ also then according to
\lem\ref{lma:EBESPrecedence}: $j \leq h$, but $h \leq j$, then $h = j$. The only
possibility for $h = j$ is that $e = e'$ because otherwise $h < j$ and $j<h$,
which is a contradiction.
\end{proof}

Let $\posets{\vartheta} = \{(C, \leq_C) \mid C \in \configurations{\vartheta}\}$
denotes the set of posets of the EBDC $\vartheta$. We show that the
transitions of a EBDC $\vartheta$ can be extracted from its posets.

\begin{theorem}
\label{thm:EBDCTrFromPosets}
Let $\vartheta$ be a EBDC and $(C,\csD), (C',\csD') \in \reachables{\vartheta}$
with $C\subseteq C'$.\\
Then
$\left(\forall e,e'\in C'\logdot e\neq e'\land e'\leq_{C'} e\implies e'\in
C \right)$ holds iff $\transDC{(C,\csD)}{(C'\csD')}$.
\end{theorem}

\begin{proof}
First, because $C'$ is a configuration it is conflict free. Now let us assume
$\forall e,e'\in C'\logdot e\neq e'\land  e'\leq_{C'} e\Rightarrow e'\in C$, we
now show that all the conditions of \defi18 in \cite{dynamicCausality15} hold
for $(C,\csD)$ and $(C',\csD')$. Condition~1,
$C\subseteq C'$, holds by assumption. Conditions~4,
5, 6, and
7 follow immediately from \lem\ref{lma:SSDCstateProp}. Condition~8 follows from 
\defi\ref{def:SingleStateDC} of SSDC, since $ \vartheta $ is an EBDC 
which is a subclass of SSDC. To prove Condition~3, let
$f\in(C'\setminus C)$, then we have from \lem\ref{lma:SingleCausalState}
$\cs[]{f}=\left(\ic{f}\cup\ac{C}{f}\right)\setminus \left(\dc{C}{f}\cup C\right)$. 
Assume $\cs[]{f} \neq
\emptyset$, \ie $ \exists f'\in \cs[]{f}$. So either $f' \in \ic{f}$ or $f \in
\cup\ac{C}{f}$. We can ignore the case that $f' \in \ac{C}{f}$, because in EBDC the added
causality for $f$ can only be $f$, which would make $f$
impossible, but this cannot be the case since $f \in C'$.
So let us consider the remaining option: $f'\in \ic{f}$. Then $f'\leq_{C'}
f$ by the definition of $\leq_{C'}$. Then by assumption, $f'\in C$ and
therefore $f'\notin \cs[]{f}$, which is a contradiction. Then $\forall
f\in(C'\setminus C) \logdot \cs[]{f} = \emptyset$. For Condition~9 
we show $\forall t,m \in C'\setminus C\logdot\forall c\in E \logdot \addcause{c}{m}{t}\implies c\in C$. 
The only growing causality is of the form $\addcause{c}{m}{c}$ and according to 
\defi\ref{def:EBDCLposets}, $\addcause{c}{m}{c}$ means $c\leq_{C'} m$,
then $c\in C$.

Let us now assume $\transDC{(C,\csD)}{(C',\csD')}$, and $e,e'\in C'$ with $e\neq
e'$ and $e'\leq_{C'} e$, so by \lem\ref{lma:EBESPrecedence} it follows $e'\in C$.
\end{proof}

The following defines a translation from an EBESs into an EBDC, which
is proved in \lem\ref{lma:EmdeddingIsEBDC} to be an EBDC. Furthermore this
translation preserves posets.
Figure~\ref{fig:exampleEBES2DCES} provides an example, where conflicts with
impossible events are dropped for simplicity.

\begin{definition}
\label{def:EBES2DCES}
Let $\xi = \left( E, \leadsto, \mapsto, l \right)$ be an EBES. Then 
 $\dces{\xi} = (E', \#', \rightarrow, \shrinkingCausality,
\growingCausality)$ such that:
\begin{inparaenum}
	\item $E', \rightarrow, \shrinkingCausality$ are
defined as in \ref{def:DESintoSES}
	\item $\#' = \{(e,e') \mid e \leadsto e' \land e'\leadsto
e\} \cup \{(x_i,x) \mid x \in X_i \}$
	\item \label{eq:EBES2DCESGrowing} $\growingCausality = \{(e,e',e) \in E^3
	\mid e\leadsto e' \land \neg (e' \leadsto e)\}$.
\end{inparaenum}
\end{definition}

\begin{figure}[tb]
	\centering
	\begin{tikzpicture}
		% figure a
		\event{a1}{-1}{1}{above}{$ a $};
		\event{a2}{-0.2}{1}{above}{$ b $};
		\event{a3}{-0.2}{0}{below left}{$ c $};
		\event{a5}{.6}{0}{right}{$ d $};
		\draw[enablingPES] (a1) edge (a3);
		\draw[enablingPES] (a2) edge (a3);
		\draw[enablingPES] (a3) .. controls (-0.9, 0.3) .. (a1);
		\draw[thick]	(-0.5, 0.37) -- (-0.2, 0.5);
		\draw[conflictPES] (a1) edge (a2);
		\draw[conflictEBES] (a3) edge (a5);
		\node (a) at (-1.7, 1) {(a)};
		% figure b
		\event{b1}{3}{1}{above}{$ a $};
		\event{b2}{3.8}{1}{above}{$ b $};
		\event{b3}{4.3}{0.7}{below right}{$ e_1 $};
		\event{b4}{3.8}{0}{below left}{$ c $};
		\event{b6}{5}{0}{below}{$ d $};
		\event{b5}{3}{0}{left}{$ e_2 $};
		\draw[enablingPES] (b5) edge (b1);
		\draw[enablingPES] (b3) edge (b4);
		\draw[enablingPES] (b3) .. controls (4.8, 1.2) and (3.8, 1.2) .. (b3);
		\draw[enablingPES] (b4) .. controls (4.3, .3) and (4.3, -.5) .. (b4);
		\draw[enablingPES] (b5) .. controls (2.5, -.5) and (3.5, -.5) .. (b5);
		\draw[adding]		(b6) -- (4.2 , 0);
		\draw[conflictPES]  (b1) edge (b2);
		\draw[dropping]		(3.1, .4) -- (b4);
		\draw[dropping]		(4, .4) -- (b1);
		\draw[dropping]		(4, .4) -- (b2);
		\node (b) at (2.3, 1) {(b)};
	\end{tikzpicture}
	\caption{An EBES and its poset-equivalent DCES.}
	\label{fig:exampleEBES2DCES}
\end{figure}
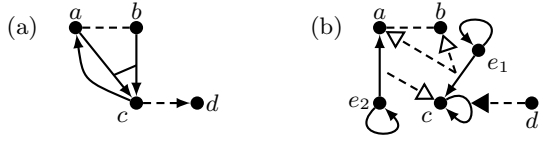

\begin{lemma}
\label{lma:EmdeddingIsEBDC}
Let $\xi$ be an EBES. Then $\dces{\xi}$ is an EBDC.
\end{lemma}

\begin{proof}
First $\dces{\xi}$ is a DCES. The definition of $\rightarrow$ in \defi\ref{def:DESintoSES} ensures Conditions~
16(1) and
16(2) in \cite{dynamicCausality15}. According to the definition of
$\shrinkingCausality$ in \defi\ref{def:DESintoSES}, the only dropped causes are the
fresh events, which cannot be added by $\growingCausality$ according to
\defi\ref{def:EBES2DCES}(\ref{eq:EBES2DCESGrowing}). So Condition~16(3) also holds.

Second, $\dces{\xi} $ is a SSDC, since the only dropped events are the
fresh ones which are never added by $\growingCausality$, so \defi\ref{def:SingleStateDC} holds.

Third,
$\dces{\xi}$ is a EBDC. \defi\ref{def:EBDC}(\ref{eq:EBDCOnlyDisabling})
holds by definition. 
Bundle members in $\xi$ mutually disable each other, then according to
the definition of $\#'$ Condition~\ref{def:EBDC}(\ref{eq:NoCausalAmbiguity})
holds. Therefore $\dces{\xi} $ is a EBDC.
\end{proof}

Before comparing an EBES with its translation according to posets, we make use
of the following lemma.

\begin{lemma}
\label{thm:EBESandEBDCconfigEqui}
Let $\xi = \left(E, \leadsto, \mapsto\right)$ be an EBES. Then $\configurations{\xi} = \configurations{\dces{\xi}}$.
\end{lemma}

\begin{proof}
First, $\forall c \subseteq E \logdot c\in \configurations{\xi}
\Longrightarrow c\in \configurations{\dces{\xi}}$. According to \S~2.2 in \cite{dynamicCausality15},
$c\in \configurations{\xi}$ means there is a trace $t = e_1,\ldots,e_n$ in $\xi$
such that $c = \bar{t}$. Let us prove that $t$ corresponds to a transition
sequence in $\dces{\xi}$ leading to $c$. \ie let us prove that there exists a
transition sequence $\transDC{(\emptyset = c_0, cs_0)} {\ldots} \transDC {}{(c_n = c,
cs_n)}$ such that $c_i = c_{i-1} \cup \Set{e_i}$ for $1 \leq i \leq n$, and
$cs_i$ is defined according to \lem\ref{lma:SingleCausalState}. This means we have
to prove that $\transDC{(c_{i-1}, cs_{i-1})} {(c_i, cs_i)}$ for $1 \leq i \leq
n$.

$c_i$ is conflict-free since it is a configuration in $\xi$ which means that it does not contain any mutual
disabling according to trace definition in \S~2.2. Second, it is clear
that $c_{i-1} \subseteq c_i$ by definition. Next, let us prove that
$\forall e\in c_i\setminus c_{i-1} \logdot cs_{i-1} = \emptyset$, \ie
$\left(\ic{e} \cup \ac{c_{i-1}}{e})\right) \setminus \dc{c_{i-1}}{e} \subseteq c_{i-1}$ according
to \lem\ref{lma:SingleCausalState}.
$\ic{e}$ contains only fresh events according to the definition of $\dces{\xi}$,
and the members of bundles $X_i\mapsto e$ are droppers of these fresh events.
But since each of these bundles is satisfied, then each of these fresh events in
$\ic{e}$ is dropped. Furthermore, there cannot be added causality in $\dces{\xi}$
for $e$ except for $e$ itself which makes it an impossible event, but it is not an impossible event since it occurs in a
configuration. Therefore $\left(\ic{e} \cup
\ac{c_{i-1}}{e}\right) \setminus \left(\dc{c_{i-1}}{e} \cup c_{i-1}\right)=\emptyset$ for
all $e \in c_i$ and all $1\leq i \leq n$. On the other hand, conditions
4, 5, 6, and 7 of
Def. 18 in \cite{dynamicCausality15} hold according to \lem\ref{lma:SSDCstateProp}.
Condition 8 of \defi18 holds by \defi\ref{def:SingleStateDC}.
Since in the transition $\transDC {(c_{i-1}, cs_{i-1})} {(c_i, cs_i)}$, only one
event --namely $e_i$-- occurs, then \defi18(9) also holds.

In that way we proved that $\configurations{\xi}\subseteq
\configurations{\dces{\xi}}$. In a similar way, and with the help of
\lem\ref{lma:SSDCConfInTrans}, we can prove that
$\configurations{\dces{\xi}}\subseteq \configurations{\xi}$ which means that $\configurations{\xi}= \configurations{\dces{\xi}}$.
\end{proof}

\begin{lemma}
\label{lma:EBESintoEBDC}
For each EBES $\xi$ there is a DCES, namely $\dces{\xi}$, such that
$\posets{\xi} = \posets{\dces{\xi}}$.
\end{lemma}

\begin{proof}
First, $\forall p \in \posets{\xi}. p \in \posets{\dces{\xi}}$. Let $p = (C,
\leq_C)$, then $C\in \configurations{\xi}$ by the definition of posets of EBESs.
Then according to \theo\ref{thm:EBESandEBDCconfigEqui}: $C \in
\configurations{\dces{\xi}}$. On the other hand, let $\leq_C'$ be the partial
order defined for $C$ in $\dces{\xi}$ as in \defi\ref{def:EBDCLposets}. This
means that we should prove that $\leq_C = \leq_C'$. But since $\leq_C, \leq_C'$
are the reflexive and transitive closures of $\prec_C, <_C$ respectively, then it is
enough to prove that $\prec_C = <_C$. In other words we have to prove
$\forall e,e'\in C \logdot e \prec_C e' \Leftrightarrow e'<_C e$.

Let us start with $e \prec_C e'\Longrightarrow e<_C e'$. According to 
\S~2.2 in \cite{dynamicCausality15} $e \prec_C e'$ means
$\exists X\subseteq E \logdot e \in X \mapsto e' \lor e\leadsto e'$. If $\exists X\subseteq E \logdot e \in X \mapsto e'$ then
$\exists c\in E' \logdot \drops{c}{e}{e'}$ by the definition of $\dces{\xi}$ \defi\ref{def:EBES2DCES}. This means $e<_C$ according to the definition of $<_C$ \defi\ref{def:EBDCLposets}.
If $e \leadsto e'$ then $ \neg e' \leadsto e$ since otherwise $e,e'$ are in conflict. This means
$\addcause{e}{e'}{e}$ according to \defi\ref{def:EBES2DCES}, which means $e<_C e'$ according to \defi\ref{def:EBDCLposets}.

Let us consider the other direction: $e <_C e'\Longrightarrow e\prec_C e'$.
$e<_C e'$ means $\exists c\in E' \logdot \drops{c}{e}{e'} \lor
\addcause{e}{e'}{e}$ according to the definition of $<_C$ in \defi\ref{def:EBDCLposets}. The third option where $e\rightarrow e'$ is rejected
since the only initial causes that exist in $\dces{\xi}$ are the fresh
impossible events. If $\exists c\in E' \logdot \drops{c}{e}{e'}$ then $\exists X \subseteq E \logdot e\in X \leadsto
e'$ according to the definition of $\shrinkingCausality$ in $\dces{\xi}$. This
means $e \prec_C e'$ by the definition of $\prec_C$ in \S~2.2. If on the
other hand $\addcause{e}{e'}{e}$ then $e \leadsto e'$ according to the definition of
$\dces{\xi}$, which means $e \prec_C$ in \S~2.2.

In that way we have proved that $\prec_C = <_C$, which means that $\leq_C =
\leq_C$. In a similar way we can prove that $\forall p \in \posets{\dces{\xi}}. p
\in \posets{\xi}$, which means $\posets{\xi} = \posets{\dces{\xi}}$.
\end{proof}

\begin{lemma}
\label{lma:DCESninEBES}
There is a DCES such no EBES with the same configurations exits.
\end{lemma}
\begin{proof}
We consider the embedding $\emb{\sigma_\xi}$ (\cf \fig\ref{fig:counterExamples}) of the SES $\sigma_\xi$, which models disjunctive causality. According to \defi13, because $\neg(a\#c)$ and $\ic{a}=\ic{c}=\emptyset$, it holds $\transS{\emptyset}{\{a,c\}}$ and so $\{a,c\}\in\configurations{\sigma_\xi}$. Further there is no transition $\transS{\emptyset}{\{b\}}$, because $\ic{b}=\{a\}$, but there are transitions $\transS{\{a\}}{\{a,b\}}$ and $\transS{\{c\}}{\{c,b\}}$, because $\ic{b}\setminus\dc{\{a\}}{b}\subseteq\{a\}$ ($\ic{b}\setminus\dc{\{c\}}{b}\subseteq\{c\}$ resp.). The transitions are translated to the embedding according to \lem\ref{lma:SESinDCES} and \defi\ref{def:DCESConfigs} the same holds for the configurations. 

If we now assume there is a EBES $\xi$ with the configurations $\emptyset,\{a\},\{c\},\{a,c\},\\\{a,b\},\{b,c\}$ and $\{a,b,c\}$ then according to \defi5 in \cite{dynamicCausality15} because there is no configuration $\{b\}$ there must be a non-empty bundle $\buEn{X}{b}$ and caused by the the configurations $\{a,b\},\{b,c\}$ this bundle $X$ must contain $a$ and $c$. Now the stability condition of \defi5 implies $a\disa c$ and $c\disa a$, so $a$ and $c$ are in mutual conflict contradicting to the assumption $\{a,c\}\in\configurations{\xi}$. Thus there is no EBES with the same configurations as $\emb{\sigma_\xi}.$
\end{proof}

Theorem~10 in \cite{dynamicCausality15} states:
\begin{quote}
DCESs are strictly more expressive than EBESs.
\end{quote}

\begin{proof}[Proof of Theorem~10 in \cite{dynamicCausality15}]
Follows directly from \lems\ref{lma:DCESninEBES} and \ref{lma:EBESintoEBDC}.
\end{proof}

\bibliographystyle{plain}
\bibliography{dynamicCausality}

\end{document}